\documentclass[10pt]{article}
\usepackage{amssymb}
\usepackage{amsmath}
\usepackage{latexsym}
\usepackage{enumerate}
\usepackage{amsthm}
\usepackage{pspicture,pstricks}
\usepackage{float}
\usepackage{pgffor}
\usepackage{epsf,amsmath,amssymb,latexsym}
\usepackage{float}
\usepackage{textpos}
\usepackage{ragged2e}
\usepackage{graphicx}
\usepackage{pstricks}
\usepackage{pst-plot}
\usepackage{pstricks}
\usepackage{pst-node}
\usepackage{pst-tree}
\usepackage{pstricks-add}

\title{On The Center Sets and Center Numbers of Some Graph Classes }
\author{\small {Ram~Kumar R\thanks{On deputation from N.S.S.College, Cherthala, and is supported by University Grants Commission, Govt. of India under their FDP Scheme. }}\\
        \small{Department of Computer~Applications}\\
        \small{Cochin University of Science and Technology}\\
        \small{ Kochi-682022, India}\\
        \small{ e-mail:ram.k.mail@gmail.com} \and
        \small{Kannan Balakrishnan}\\
           \small{Department of Computer~Applications}\\
           \small{Cochin University of Science and Technology}\\
           \small{ Kochi-682022, India}\\
           \small{ e-mail:mullayilkannan@gmail.com}\and
           \small{Manoj Changat}\\
           \small{Department of Futures Studies}\\
           \small{University of Kerala}\\
           \small{ Kariyavattom, Trivandrum - 695581, India}\\
           \small{e-mail:mchangat@gmail.com}\and
           \small{A.Sreekumar}\\
           \small{Department of Computer~Applications}\\
           \small{Cochin University of Science and Technology}\\
            \small{ Kochi-682022, India}\\
           \small{e-mail:sreekumar@cusat.ac.in}
           \and
           \small{Prasanth G. Narasimha-Shenoi}\\
          \small{Department of Mathematics}\\
           \small{Government College, Chittur}\\
           \small{Palakkad-678104, India}\\
           \small{e-mail:prasanthgns@gmail.com}
           }
           \date{ }

\newtheorem{theorem}{Theorem}
\newtheorem{lemma}{Lemma}
\newtheorem{remark}{Remark}
\newtheorem{proposition}{Proposition}
\newtheorem{corollary}[theorem]{Corollary}

\begin{document}
\maketitle
\begin{abstract}
For a set $S$ of vertices and the vertex $v$ in a connected graph $G$, $\displaystyle\max_{x \in S}d(x,v)$ is called the $S$-eccentricity of $v$ in $G$. The set of vertices with minimum $S$-eccentricity is called the $S$-center of $G$. Any set $A$ of vertices of $G$ such that $A$ is an $S$-center for some set $S$ of vertices of $G$ is called a center set. We identify the center sets of certain classes of graphs namely, Block graphs, $K_{m,n}$, $K_n-e$, wheel graphs, odd cycles and symmetric even graphs and enumerate them for many of these graph classes. We also introduce the concept of center number which is defined as the number of distinct center sets of a graph and determine the center number of some graph classes. 
 \end{abstract}

  \section{Introduction}
 centrality is one of the fundamental notions in graph theory which has established close connection between graph theory and various other areas like Social networks, Flow networks, Facility location problems etc. The main objective of any facility location problem is to identify the  location of a facility for a community or set of customers such that the distance between the location and the community or customers is minimized.  This leads to the standard notion of graph centers, which is widely studied and still continues to be an important branch in metric graph theory. The concept of centrality has recently gained popularity  in large networks like  where the aim is to identify "important actors" and this is done using different centrality concepts such as degree, closeness and betweenness centrality ~\cite{wasserman1994social}. The center of a graph consists of those vertices with  minimum eccentricity, where eccentricity of a vertex is the maximum distance of the vertex among the set of all vertices.\\ 
The problem of finding the center of a graph has been studied by many authors since the nineteenth century beginning with the classical result due to Jordan \cite{jordan1869assemblages} that the center of a tree consists of a single vertex or a pair of adjacent vertices. The graph center problem is interesting from both a structural and an algorithmic point of view. Harary and Norman in \cite{harary1953dissimilarity} proved that the center of a connected graph lies with in a block of the graph. Buckley et. al in \cite{buckley1981graphs}, examined the problem of embedding a graph $H$ as the center of a supergraph $G$ and showed that for each graph $H$ with $n>9$ nodes and an integer $k>n+1$ there exists a $k$-regular graph $G$ having the center $H$. Proskurowski described the centers of maximal outer planar graphs and $2$-trees, \cite{proskurowski1980centers,proskurowski2011centers}. Laskar and Shier in \cite{laskar1983powers} proved that for a connected chordal graph the center always induces a connected subgraph, further Chepoi in \cite{chepoi1988centers} characterized the centers of chordal graphs. Also refer  Chepoi~\cite{chepoi1994linear}, for linear time algorithms for centers in chordal graphs, Soltan and Chepoi~\cite{soltan1984d} and  Chang~\cite{chang1991centers} for more problems on centers of chordal graphs.\\

Slater in \cite{slater1978centers} generalized the concept of center of a graph to center of a profile ( an arbitrary subset) of the vertex set of the graph. More formally, for any subset $S $ of $V$ in the graph $G=(V,E)$ the $S-eccentricity$, $e_{G,S}(v)$ (in short $e_S(v)$) of a vertex $v$ in $G$ is $\max\limits_{\substack{x \in S}}(d(v,x))$.  The $S$-center of $G$ is $C_{S}(G)= \{ v \in V |e_{S}(v) \le e_{S}(x)\, \forall x \in V\}$. Chang in \cite{chang} studied the $S$-center of distance hereditary graphs and proved that the $S$-center of a distance hereditary graph is either a connected graph of diameter 3 or a cograph. He also proved that for a bipartite distance hereditary graph the $S$-center is either a connected graph of diameter $\le 3$ 3 or an independent set. Various centrality notions in graphs is discussed by Buckley and Harary  in their book \cite{buha-94 }.  

In this paper, we continue the studies with the centers of profiles $S$  of vertices on several classes of graphs. When a subset $S$ of a vertex set of a graph has a center, we call this the center set of the set $S$. We characterize the center sets of several classes of graphs, namely block graphs, complete graphs, complete bipartite graphs, symmetric even graphs, etc.  We also characterize those class of graphs for which it doesn't contain a center set, namely, center critical graphs.

Several metric related invariants has been studied in graphs, for example, the diameter, radius, genus etc. being one of the basic parameters. In this paper, we introduce a new invariant related to the center of a graph, namely the \emph{center number} of a graph as the number of exact center sets in a graph and enumerate the center number of some classes of graphs, which will give an idea of the number of possible emergency facility locations in the such networks. In the next paragraph, we fix the notations and define the basic concepts used in this paper.

We consider only finite simple undirected connected graphs. For the graph $G$, $V(G)$ denotes its vertex set and $E(G)$ denotes its edge set. If the circumstances are clear, we use $V$ and $E$ for $V(G)$ and $E(G)$ respectively. For two vertices $u$ and $v$ of $G$, distance between $u$ and $v$ denoted by $d(u,v)$, is the length of the shortest $u-v$ path. The degree of a vertex $u$, denoted by $deg(u)$ is the number of vertices adjacent to $u$. A vertex $v$ of a graph $G$ is called a cut-vertex if $G-v$ is no longer connected. Any maximal induced subgraph of $G$ which does not contain a cut-vertex is called a \emph{block} of $G$. A graph $G$ is a \emph{block graph} if every block of $G$ is complete.  The \emph{eccentricity} $e(u)$ of a vertex $u$ is $\max\limits_{\substack{v \in V(G)}} d(u,v)$. A vertex $v$ is an \emph{eccentric vertex} of $u$ if $e(u)= d(u,v)$.  A vertex $v$ is an \emph{eccentric vertex} of $G$ if there exists a vertex $u$ such that $e(u)=d(u,v)$.  The diameter of the graph $G$, $diam(G)$, is $\max\limits_{\substack{u \in V(G)}} e(u)$ and the radius, $rad(G)$, is $\min\limits_{\substack{u \in V(G)}} e(u)$. A graph $G$ is called \emph{even} if for each vertex $u$ of $G$ there is a unique eccentric vertex $\bar{u}$, such that $d(u,\bar{u}) = diam(G)$. An even graph $G$ is called \emph{balanced} if $deg(u) = deg(\bar{u})$ for each $u \in V$, \emph{harmonic} if $\bar{u} \bar{v} \in E$ whenever $uv \in E $ and \emph{symmetric} if $d(u,v) + d(u,\bar{v}) = diam(G)$ for all $u,v \in V$.  The interval $I(u,v)$ between vertices $u$ and $v$ of $G$ consists of all vertices which lie in some shortest path between $u$ and $v$. A vertex $u$ of a graph $G$ is called a \emph{universal vertex} if $u$ is adjacent to all other vertices of $G$. An $S\subseteq V$  is a \emph{dominating set} in  $G$ if every vertex in $V \setminus S$ is adjacent to a vertex in $S$. The set of all vertices adjacent to $x$ in a graph $G$, denoted by $N(x)$, is the \emph{neighborhood} of the vertex $x$. For an $S \subseteq V$, a vertex $x \in S$ is called an \emph{interior vertex} if $N(x) \subseteq S$. For an $S \subseteq V$, \emph{neighborhood} of $S$ denoted by $N(S)=\bigcup\limits_{\substack{u \in V}} N(u)$. An $S\subseteq V$ is called a \emph{boundary set} of $G$ if does not contain any interior vertices.
For a graph $G$, an $A \subseteq V$ is defined to be a \emph{Center set} if there exists an $S \subset V$ such that $C_{S}(G)=A $. \\

The paper is divided as follows. In Section \ref{centersets} we identify the center sets of some familiar classes of graph such as block graphs, complete bipartite graphs, wheel graphs, odd cycles, symmetric even graphs etc. 
In Section~\ref{enumeration}, an enumeration of the center number of various classes of graphs is discussed and conclude the paper with some remarks. Before we go into the main contents of the paper, we characterize a class of graphs called \emph{center critical graphs}, which are graphs such that, for all proper subsets $S$ of $V$, $C_{S}(G)\not=C(G)$. Further, a graph $G$ is a \emph{unique eccentric vertex} graph ( in short, a $UEV$), if every vertex of $G$ has a unique eccentric vertex. The unique eccentric vertex of a vertex $u$ is denoted by $\bar{u}$. Note that a unique eccentric vertex graph $G$ become a self centered graph ( Self centered graph is a graph in which every vertex has the same eccentricity) if and only if each vertex of $G$ is an eccentric vertex of some other vertex in $G$ has been proved by Parthasarathy et al. in~\cite{parthasarathy1983unique}. It may also be observed that a $UEV$ graph need not be self centered and a self centered graph need not be a $UEV$ graph. For example, all even paths are $UEV$ graphs, but are not self centered and the complete graph $K_n$ and complete bipartite graph $K_{n,n}$ are self centered, but are not $UEV$ graphs. But we give a characterization of central critical graphs as precisely those graphs which are both self centered and $UEV$.

\begin{theorem}
A graph $G$ is center critical if and only if $G$ is both self centered and a $UEV$ graph.
\end{theorem}

\begin{proof}

 Let $G$ be a center critical graph having vertex set $\{v_1,\ldots,v_n\}$. First we shall prove that for every $v_i \in V$ there exists a $v_j \in V$ such that $v_i$ is the unique eccentric vertex of $v_j$. Assume the contrary. Here we shall take two cases.\\
 \textbf{Case-I}: Let there exist a vertex, say $v_k$, such that $v_k$ is not an eccentric vertex of any vertex.
Let $S=V \setminus \{v_k\}$. Then for every vertex $v_i$ of $G$, $e_S(v_i)=e(v_i)$ since the eccentric vertices of $v_i$ are in $S$. Since the eccentricities of none of the vertices change, $C_S(G)=C(G)$ contradicting our assumption that $G$ is center critical. Hence every vertex of $G$ is an eccentric vertex. \\
\textbf{Case-II}: Let $v_k$ be such that when ever $v_k$ is an eccentric vertex of $v_\ell$  then there exists a vertex $v_k'$ such that $v_k'$ is also an eccentric vertex of $v_\ell$. 
 Again take $S=V \setminus \{v_k\}$. Since every vertex $v_\ell$ that has $v_k$ as an eccentric vertex has another eccentric vertex, we have $e_S(v_k)=e(v_k)$. As above we get that $C_S(G)=C(G)$, a contradiction.
  That is, we have proved that each vertex $v_i$, $1 \le i \le n$ is a unique eccentric vertex of a vertex, say $v_i'$, where $v_i'=v_j$ for some $j$, $1 \le j \le n$. Since $\{v_1',\ldots,v_n'\}=V$ and each $v_i'$ has a unique eccentric vertex each vertex of $G$ has a unique eccentric vertex. Now its is also obvious that every vertex is an eccentric vertex. Therefore by the result of Parthasarathy, $G$ is self centered.

Conversely assume that $G$ is both self centered and unique eccentric vertex graph, and let $rad(G)=r$. By the result of Parathasarathy,  every vertex of $G$ is an eccentric vertex. Therefore for every $x \in V$ there exists a $y \in V$ such that $x=\bar{y}$. Let $S \subseteq V$ and $x \in V \setminus S$. Then $e(y)=r$ and since $\bar{y}=x \in V \setminus S$, $e_S(y) <r$. Let $z \in S$. Then $e_S(\bar{z})=r$. Hence $C_S(G) \neq V$ which shows that $G$ is center critical.

\end{proof}

\section{Center Sets of Some Graph Classes}\label{centersets}
Prior to identifying the center sets of various classes of graphs we recall the following lemma by Harary et.al in \cite{harary1953dissimilarity}.
\begin{lemma}[Lemma 1 of \cite{harary1953dissimilarity}]\label{harary}
The center of a connected graph $G$ is contained in a block of $G$.
\end{lemma}
We generalize this lemma to any $S$-center of a graph and the proof is almost similar to the proof given there.
\begin{theorem}
Any $S$-center of a connected graph $G$ is contained in a block of $G$.
\end{theorem}
\begin{proof}
For an $S \subseteq V $, assume that $C_S(G)$ lies in more than one block of $G$. Then $G$ contains a vertex $v$ such that $G-v$ contains atleast two components, say, $G_1$ and $G_2$, each of which contains a vertex belonging to $C_S(G)$. Let $u$ be the vertex of $S$ such that $d(u,v)=e_S(v)$ and $P$ be the shortest $u-v$ path. Then $P$ does not intersect at least one of $G_1$ and $G_2$, say $G_1$. Let $w$ be the vertex of $G_1$ such that $w \in C_S(G)$. Then $v$ belong to the shortest $w-u$ path and hence \\
$e_S(w)  \ge d(w,u)=d(w,v)+d(u,v) \ge 1+ e_S(v)$ contradicting the fact that $w \in C_S(G)$. Thus for any $S \subseteq V$, $C_S(G)$ lies in a single block of $G$.
\end{proof}
\begin{proposition}
 Let $G$ be a block graph with vertex set $V$ and blocks $B_1,\ldots ,B_r$. For $1 \le i \le r$, let $V(B_i)=V_i$. The center sets of $G$ are singleton sets $\{v\}, v \in V(G)$ and $V_i$ for $1\le i \le r$.
 \label{lem1}
\end{proposition}
\begin{proof}
If $S=\{v\}$, then $e_S(v)=0\leq e_S(x)$ for all $x\in V$.  Therefore $C_{\{v\}}(G)=\{v\}$. Hence $\{v\}$, where $v\in V$ are all center sets.

Let $S$ be a proper subset of $V_i$, $1 \le i \le r$  containing atleast two elements . Hence $e_S(x)=1$ for every $x\in V_i$ and $e_S(x)>1$ for all $x\in V-V_i$.  So $C_S (G)=V_i$. Therefore each $V_i$, $1 \le i \le r$ is a center set.

Consider $S \subseteq V(G)$ containing atleast $2$ elements from $2$ different blocks, and let $x$ be a cut vertex of $G$ with $e_S(x)=k$. Also assume that $d(x,v)=k$ where $v \in S$. Let $P=x=x_0x_1\ldots x_rx_{r+1}\ldots x_k=v$ be the shortest $x-v$ path.  See that $e_{S}(x_1)=k-1$. Since the eccentricities will never decrease to zero, we can find two vertices in $P$ (may be the identical) say $x_r$, and $x_{r+1}$ so that $e_{S}(x_r)=e_{S}(x_{(r+1)})=k-r$. Then for every vertex $y$ in the block containing $x_r$ and $x_{r+1}$, $e_{S}(y)=k-r$ and as we move away from this block the $S$-eccentricity increases.  Hence the $S$-center of $G$ is the block containing $x _r$ and $x_{r+1}$.

Now let $e_{S}(x_r)=k-r$ and $e_{S}(x_{r+1})=k-r+1$.  Then for every $y$ other than $x_r$ in the block containing $x_r$ and $x_{r+1}$, $e_{S}(y)=k-r+1$ and as we move away from this block the $S$-eccentricity increases. Therefore $S$-center of $G$ is $x_r$. Hence the center sets of block graphs are $\{v\}$, $v \in V(G)$ and $V_i, 1 \le i \le r$.
\end{proof}
As a consequence of the Proposition \ref{lem1}, we have the following corollaries. Note that the Corollary \ref{cor1}, is a theorem of Slater in \cite{slater1978centers}.
\begin{corollary}
The center sets of the complete graph $K_n$ with vertex set $V$ are $\{u\},u \in V$ and the whole set $V$.
\label{cor1}
\end{corollary}
\begin{corollary}[Theorem $4$ of \cite{slater1978centers}]
The center sets of a tree $T=(V,E)$ are $\{u\},u \in V$, and $\{u,v\}, uv \in E$
\label{cor2}
\end{corollary}
\begin{corollary}
The induced subgraphs of all center sets of a block graph are connected.
\end{corollary}

Now we shall find the center sets of some simple classes of graphs such as complete bipartite graphs, $K_n-e$, Wheel graphs, etc. First we identify the center sets of bipartite graphs $K_{m,n}$, $m,n >1$. When $m$ or $n$ is 1, $K_{m,n}$ is a tree whose center sets have already been identified.
\begin{proposition}
Let $K_{m,n}$ be a complete bipartite graph with bipartition $(X,Y)$ where $|X|=m>1$ and $|Y|=n>1$. Then the center sets of $K_{m,n}$ are
\begin{enumerate}
\item $V=X \cup Y$
\item $X$
\item $Y$
\item $\{v\}, v \in V$
\item $\{x,y\}, x \in X, y \in Y$.
\end{enumerate} \label{kmn}
\end{proposition}
\begin {proof}
First we shall show that each of the sets described in the theorem are center sets.
\begin{enumerate}
\itemsep-7pt
\item Let $A=A_1 \cup A_2$ where $A_1 \subseteq X$, $A_2 \subseteq Y$, $|A_1| >1$ and $|A_2|>1$. For all $v \in V$, $e_A(v)=2$ and therefore $C_A(K_{m,n})=V$.\\
\item Take $A_2 \subseteq Y$ with $|A_2|>1$. Then for each $x \in X$, $e_{A_2}(x)=1$ and for each $y \in Y, e_{A_2}(y)=2$. Hence $C_{A_2}(K_{m,n})=X$.\\
\item Take $A_1 \subseteq X$ with $|A_1|>1$. Then for each $x \in X$, $e_{A_1}(x)=2$ and for each $y \in Y, e_{A_1}(y)=1$. Hence $C_{A_1}(K_{m,n})=Y$.\\
\item Take $A= \{x\} \cup A_2$ where $x \in X$, $A_2 \subseteq Y$ and $|A_2|>1$. Then $e_{A}(x)=1$ and $e_{A}(v)=2$ for all $v \neq x$. Hence $C_{A}(K_{m,n})=\{x\}$. Similarly taking  $A= \{y\} \cup A_1$ where $x \in Y$, $A_1 \subseteq X$ and $|A_1|>1$ we get $C_{A}(K_{m,n})=\{y\}$. Hence for every $v \in V$, $\{v\}$ is a center set.\\
\item Take $A=\{x,y\}$ where $x \in X$ and $y \in Y$. Then $e_{A}(x)=1$, $e_{A}(y)=1$ and for all other $v \in V$, $e_{A}(v)=2$. Hence $C_{A}(K_{m,n})=\{x,y\}.$
\end{enumerate}
Hence we have got that all the sets described in the theorem are center sets. Also we have found out the center sets of all types of subsets of $V(K_{m,n})$ except $\{v\}$, $v \in V$. But $C_{\{v\}}( K_{m,n})=\{v\} $. Therefore center sets of $K_{m,n}$ are precisely those given in the theorem.
\end{proof}
\begin{proposition}
For the graph $K_n-e(=xy)$, the center sets are
\begin{enumerate}
\itemsep-9pt
\item $\{v\}$, $v \in V$\\
\item $V \setminus \{x\}$\\
\item $V \setminus \{y\}$\\
\item $V \setminus \{x,y\}$\\
\item $V$
\end{enumerate}\label{kn-e}
\end{proposition}
\begin{proof} As in Proposition~\ref{kmn}, initially we prove that all the sets described in the theorem are center sets.
\begin{enumerate}
\itemsep-7pt
\item For each $v \in V$, $C_{\{v\}}( K_{n}-e)=\{v\} $.\\
\item Let $S \subseteq V$ be such that $|S|>1$, $y \in S$ and $x \notin S$. Then for each $v \neq x$,  $e_{S}(v)=1$ and  $e_{S}(x)=2$. Hence $C_{S}(K_n-e)= V \setminus \{x\}$.\\
\item For $S \subseteq V$  such that $|S|>1$, $x \in S$ and $y \notin S$ we have $e_S(v)=1$ for $v \neq y$ and $e_S(y)=2$. Hence $C_{S}(K_n-e)= V \setminus \{y\}$.\\
\item Let $S \subseteq V$ be such that $x,y \in S$. Then for each $v \neq x,y$, $e_{S}(v)=1$, $e_{S}(x)=e_{S}(y)=2$ and hence  $C_{S}(K_n-e)= V \setminus \{x,y\}.$\\
\item Let $S \subseteq V$ be such that $|S|>1$, $x,y \notin S$. Then for every $v \in V$, $e_{S}(v)=1$ and therefore $C_{S}(K_n-e)= V$.
\end{enumerate}
Now we have found the centers of all types of subsets of $V$ and therefore above mentioned sets are precisely the center sets of $K_n-e$.
\end{proof}
Now we shall identify the center sets of wheel graphs. The wheel graph $W_4$ is $K_4$ and their center sets have already been identified. First we prove the case for $ n \ge 6$. The center sets of $W_5$, the only remaining case, will be given in the remark after the Proposition~\ref{wn}.

\begin{proposition}\label{wn}
Let $W_n$, $n \ge 6$, be wheel graph on the vertex set $\{v_1,\ldots,v_n\}$ where $v_n$ is the universal vertex. Then the center sets of $W_n$ are
\begin{enumerate}
\item $\{v_i\}$, $1 \le i \le n$
\item $\{v_i,v_n\}$, $1 \le i \le n-1$
\item $\{v_i,v_j,v_n\}$, \text{where} $v_iv_j \in E(C_{n-1})$
\item $\{v_i,v_j,v_k,v_n\}$ \text{where} $v_iv_j,v_jv_k \in E(C_{n-1}) $
\end{enumerate}
\end{proposition}
\begin{proof}
First we shall prove that each of the sets described above are center sets. In the proof $i+k$ (or $i-k$) means $i+k-(n-1)$(or $i-k+(n-1)$) when $i+k >n-1$(or $i-k <1$).
\begin{enumerate}
\item For  $1 \le i \le n$, $C_{\{v_i\}}(G)=\{v_i\}$.
\item Let $S=\{v_{i-1},v_i,v_{i+1}\}$. $e_S(v_i)=e_S(v_n)=1$ and $e_S(v)=2$ for all other $v \in V$ and therefore $C_S(G)=\{v_i,v_n\}$.
\item For $S=\{v_i,v_{i+1},v_n\}$, $C_{S}(G)=S=\{v_i,v_{i+1},v_n\}$.
\item For $S=\{v_i,v_n\}$, $C_S(G)=\{v_{i-1},v_{i},v_{i+1},v_n\}$.
\end{enumerate}
For all $S \subseteq V$ such that $S \neq \{v_n\}$, $e_S(v_n) =1$ and hence for all $S \subseteq V$ such that $S\neq \{v_i\}$, $1 \le i \le n-1$, $v_n \in C_S(G)$. Now, let $A$ be such that $A$ contain $v_i$ and $v_j$ such that $d_{C_{n-1}}(v_i,v_j)>2$. Let $S \subseteq V$ be such that $C_S(G)= A$ then obviously $S\neq \{v_i\}$, $1 \le i \le n$. We have $v_n \in C_S(G)$ with $e_S(v_n)=1$ Therefore $v_i$ and $v_j$ belong to $C_S(G)$ implies there exist a vertex $v_k$ in $V(C_{n-1})$ such that $d(v_i,v_k)=d(v_j,v_k)=1$ which is impossible by the choice of $v_i$ and $v_j$. Hence $v_i$ and $v_j$ of $V(C_{n-1})$ belong to a center set implies $d_{C_{n-1}}(v_i,v_j) \le 2$. Also $v_i$, $v_{i+2}$ belong to $C_S(G)$ implies $v_{i+1}$ belong to $C_S(G)$. Hence the center sets are precisely those described in the theorem.
\end{proof}

\begin{remark}

Let $\{v_1,v_2,v_3,v_4,v_5\}$ be the vertex set of $W_5$ with $v_5$ as the universal vertex. All sets of the types given in the Proposition \ref{wn} are center sets in the same manner. Since the outer cycle is of length $4$, $C_{\{v_1,v_3\}}(W_5)=\{v_2,v_4,v_5\}$ and   $C_{\{v_2,v_4\}}(W_5)=\{v_1,v_3,v_5\}$. By the arguments similar to that given in the proof of Proposition \ref{wn}, the center sets of $W_5$ are precisely
\begin{enumerate}
\item $\{v_i\}$, $1 \le i \le5$
\item $\{v_i,v_5\}$, $1 \le i \le 4$
\item $\{v_i,v_j,v_5\}$, where $v_iv_j \in E(C_4)$
\item $\{v_i,v_j,v_k,v_n\}$ where $v_iv_j,v_jv_k \in E(C_4) $
\item $\{v_1,v_3,v_5\}$, $\{v_2,v_4,v_5\}$
\end{enumerate}
\end{remark}
\begin{remark}
The induced subgraphs of all center sets of a wheel graph are connected. Infact the induced subgraphs of all centersets of any graph with a universal vertex are connected.
\end{remark}
 \begin{theorem}
  Let $C_{2n+1}$ be an odd cycle with vertex set $V=\{v_1,\ldots v_{2n+1}\}$. An $A \subseteq V$ is a center set of $C_{2n+1}$ if and only if either $A=V$ or $A$ does not contain a pair of alternate vertices. \label{odd}
 \end{theorem}
 \begin{proof}
 If $A=V$ then it is a center set namely, of itself. So assume $A \neq V$. Let $A \subset V$ be such that it contain three consecutive vertices say, $v_1,v_2,v_3$. Assume there exists an $S \subset V$ with $A =C_S(G)$. Let $d$ be the $S$-eccentricity of a vertex of $A$. Then there exists a vertex $v_i$ in $S$ such that $d(v_1,v_i)=d$. $d(v_2,v_i)=d$ implies $v_1$ and $v_2$ are the eccentric vertices of $v_i$ which means $d=n$ or $A=V$. Hence $d(v_2,v_i) \ne d$. $d(v_2,v_i)=d+1$ implies $e_S(v_2)\ge d+1$. Hence $d(v_2,v_i)=d-1$. Then there exists a vertex $v_j$ such that $d(v_2,v_j)=d$ and $d(v_1,v_j)=d-1$. Then as explained above $d(v_3,v_j)$ cannot be $d$ and therefore $d(v_3,v_j)=d+1$. This means that $e_S(v_2) \ne e_S(v_3)$. Hence any three consecutive vertices cannot be in a center set. Now, assume that $A \subset V$ is such that it contains a pair of alternate vertices and does not contain the middle vertex, say, contains $v_1$ and $v_3$ and does not contain $v_2$. Assume $A =C_S(G)$. Let $e_S(v_1)=e_S(v_3)=d$. Then $e_S(v_2)=d+1$. Let $v_i$ be a vertex in $S$ such that $d(v_2,v_i)=d+1$. Obviously $d(v_1,v_i)=d(v_3,v_i)=d$ and this implies $v_i$ is the eccentric vertex of $v_2$ or $d(v_2,v_i)=n$. But since $C_{2n+1}$ is an odd cycle either $d(v_1,v_i)=n$ or $d(v_3,v_i)=n$, a contradiction. Hence if $A$ is a center set then it cannot contain a pair of alternate vertices.\\
Conversely assume that $A$ is such that it does not contain any pair of alternate vertices of the cycle. Now take $S$ to be the set of all vertices of $C_{2n+1}$ which are eccentric vertices of vertices of $A^c$ and which are not eccentric vertices of  any of the vertices of $A$. It is obvious by the choice of $A$ that such vertices do exist. Since an antipodal vertex of atleast one of the two neighbours of each vertex of $A$ belong to $S$ and none of the antipodal vertices of any vertex of $A$ belong to $S$, for each vertex $x$ of $A$, $e_S(x)=n-1$. Since atleast one of the antipodal vertices of each vertex of $A^c$ belong to $S$, for each vertex $y$ of $A^c$, $e_S(y)=n$. Thus $A =C_S(G)$. Hence the theorem.
\end{proof}
\begin{corollary}
For the odd cycle $C_{2n+1}$, if $A$ is a center set then either $|A| \le n $ or $|A|=2n+1$.\label{oddcent}
\end{corollary}
\begin{proof}
Suppose $A$ is a center set such that $|A| < 2n+1$. To prove $|A| \le n$. Since $A$ is a center set $A$ cannot contain a pair of alternate vertices of the cycle. Let each vertex belonging to $A$ be represented by $1$ and each vertex not belonging to $A$ be represented by $0$. Thus we get a circular arrangement of $0$'s and $1$'s such that two successive $0$'s contains at most two $1$'s between them and if a pair of successive zeros contain a $1$ between them then the next pair of successive 0's does not contain a $1$ between them. From this as well as the fact that the cycle contains odd number of vertices we can conclude that $m$ $0$'s can accommodate at most $(m-1)$ $1$'s between them. Therefore if $A'$ is a center set such that $A'$ has a maximum cardinality among the center sets other than $V$ then the binary representation of $A'$ will have exactly $n+1$ 0's and hence $n$ $1$'s. In other words $|A'|=n$.
\end{proof}
\begin{corollary}
For any $m \le n$, there exists an $S \subseteq V(C_{2n+1})$ such that $|C_S(C_{2n+1})|=m$.\label{realodd}
\end{corollary}
\begin{proof}
Let $m \le n$. Take $2n+1-m$ circularly arranged $0$'s. Number these $0$'s $1,2, \ldots, 2n+1-m$. If $m$ is even put two $1$'s each between the first and the second $0$'s, third and the fourth $0$'s etc up to $(m-1)^{th}$ and the $m^{th}$ $0$'s. If $m$ is odd put two $1$'s each between the first and the second $0$'s, third and the fourth $0$'s etc., up to $(m-2)^{th}$ and the $(m-1)^{th}$ $0$'s and one $1$ between $m^{th}$ and $(m+1)^{th}$ $0$'s. In both these cases we get a circular arrangement of $0$'s and $1$'s such that it has $m$ $1$'s and does not contain a pattern of the type $101$ or $111$. That is given an $m \le n$, we can find a subset of $V(C_{2n+1})$ of size $m$ which does not contain any pair of alternate vertices. In other words given an $m \le n$, we can find a center set of size $m$.
\end{proof}

The following theorem gives the center sets of some familiar classes of graphs such as even cycles, hypercubes etc.
\begin{theorem}
Let $G$ be a symmetric even graph. An $A \subseteq V$ is a center set if and only if either $A=V$ or there does not exist an $x$ in $V$ such that $\{x\} \cup N(x) \subseteq A$. \label{symevcent}
\end{theorem}
\begin{proof}
Since symmetric even graphs are self centered $C_V(G)=V$. So assume $A \subset V$. Let $A$ be such that $A =C_S(G)$ for an $S \subset V$ and let $x \in A$. Suppose $e_S(x)=k$ with $d(x,y)=k$ where $y \in S$. If $k=diam(G)$ then $A=V$. So assume $k < diam(G)$. Then since $G$ is a symmetric even graph there exists a vertex $z$ adjacent to $x$ such that $d(y,z)=k+1$. Therefore $e_S(z)\ge k+1$ or $z \notin C_S(G)$. Hence if $A$ is a center set such that $A \subset V$, then there exists an $x$ in $A$ such that $\{x\} \cup N(x) \cap S^{c} \neq \emptyset$.\\

Conversely, suppose that $A \subset V$ satisfies the condition given in the theorem. We need to find out an $S \subseteq V$ such that $A =C_S(G)$. Since $G$ is symmetric even it is self centered and unique eccentric vertex. Let $\overline{A^c}$ denote the set of eccentric vertices of $A^c$. Let $x \in A$. Then there exists a $x'$ adjacent to $x$ such that $x' \in A^c$. Then $\overline{x'} \in \overline{A^c}$. Since $d(x',\overline{x'})= diam(G)$ and $x$ and $x'$ are adjacent $d(x,\overline{x'})= diam(G)-1$. Also since $G$ is unique eccentric vertex there does not exist an $z$ in $\overline{A^c}$ such that $d(x,z)=diam(G)$. Therefore, $e_{\overline{A^c}}(x)=diam(G)-1$ and for every $y \in A^c$, $e_{\overline{A^c}}(x')=diam(G)$. Since $G$ is self centered for every $x \in A$, $e_{\overline{A^c}}(x)=diam(G)-1$ and for every  $y \in A^c$, $e_{\overline{A^c}}(x')=diam(G)$. Therefore $C_{A^c}(G)=A$. Hence the theorem.
\end{proof}
\begin{corollary}
For the even cycle $C_{2n}$, if $A$ is a center set then either $|A| \le \lfloor\frac{4n}{3}\rfloor $ or $|A|=2n$. \label{evencent}
\end{corollary}
\begin{proof}
Suppose $A$ is a center set such that $|A| < 2n$. To prove $|A| \le \lfloor\frac{4n}{3}\rfloor$. Since $A$ is a center set $A$ cannot contain three consecutive vertices of the cycle. As in Corollary \ref{oddcent}, form a binary representation corresponding to the set $A$. Thus we get a circular arrangement of $0$'s and $1$'s such that two successive $0$'s contains at most two $1$'s between them. From this we can conclude that $m$ $0$'s can accommodate at most $2m$ $1$'s between them.
If $A'\not= V$ is a center set of maximum cardinality then the binary representation of $A'$ will have exactly $\lceil\frac{2n}{3}\rceil$ zeros and hence $2n-\lceil\frac{2n}{3}\rceil$ $1$'s.
In other words $|A'|=2n-\lceil\frac{2n}{3}\rceil=\lfloor\frac{4n}{3}\rfloor$.  Since $A'$ is a center set of maximum cardinality, we have $|A|\le \lfloor\frac{4n}{3}\rfloor$.  Hence the corollary.
\end{proof}
Next we have another corollary similar to the Corollary \ref{realodd}.
\begin{corollary}
For any $m \le \lfloor\frac{4n}{3}\rfloor$, there exists an $S \subseteq V(C_{2n})$ such that $|C_S(C_{2n})|=m$.\label{realeven}
\end{corollary}
\begin{proof}
Similar to the proof of Corollary \ref{realodd}
\end{proof}
Next we prove certain results regarding the centers of some special subsets of $V$ in symmetric even graphs. Before that we have the following propositions from \cite{key18}.
\begin{proposition}
Every harmonic even graph is balanced.
\end{proposition}
\begin{proposition}
Every Symmetric even graph is harmonic.
\end{proposition}
Combining the above two propositions we get the following proposition.
\begin{proposition}
Every Symmetric even graph is balanced. \label{symev}
\end{proposition}
\begin{theorem}
Let $G$ be a symmetric even graph and let $S \subseteq V$. Then $C_S(G)=\overline{S^c}$ if and only if $S$ is a dominating set.
\end{theorem}
\begin{proof} Assume $C_S(G)=\overline{S^c}$. Suppose $S \cup N(S) \neq V$. Then there exists an $x \in V$ such that $x \notin S$ and $x \notin N(S)$. That is $x$ and all its neighbours belong to $S^c$. Let $x_1, \ldots, x_k$ be the neighbours of $x$. By proposition \ref{symev}, $deg(u)=deg(\bar{u})$. Let $y_1,y_2, \ldots, y_k$ be the neighbours of $\bar{x}$. We have $d(x_i,\bar{x})=diam(G)-1$ for $1 \le i \le k$. Since $G$ is symmetric even there exists a vertex adjacent to $\bar{x}$, say $y_i$, such that $d(x_i,y_i)=daim(G)$ for $1 \le i \le k$. Hence $\bar{x}$ and all its neighbours belong to $\overline{S^c}$. This contradicts the condition for $\overline{S^c}$ to be a center set.\\
Conversely suppose $S \cup N(S)=V$. Let $x \in \overline{S^c}$. Then $\bar{x} \in S^c$. Since $S \cup N(S) = V$, $\bar{x} \in N(S)$. Therefore there exists an $z \in S$ such that $z$ is adjacent to $\bar{x}$. Then $d(x,z)=daim(G)-1$. $d(x,z')=diam(G)$ for some $z' \in S$ implies  both $y  \in S^c$ and $z' \in S$ are the eccentric vertices of $x$ a contradiction to the fact that the graph is unique eccentric vertex. Hence $e_S(x)=diam(G)-1$. Now let $x \notin \overline{S^c}$. Then since every vertex is an eccentric vertex, $x \in \overline{S}$ and therefore there exists a $w$ in $S$ such that $d(x,w)=diam(G)$. Thus $C_S(G)=\overline{S^c}$.
\end{proof}
For a graph $G$, let $\mathcal{DB}(G)$ denote the class of dominating boundary sets, that is, dominating sets which are also boundary sets. We have the following theorem on the centers of sets which belong to such a class of sets in a symmetric even graph.
\begin{theorem}
Let $G$ be a symmetric even graph. Let $S \subseteq V$ be such that $S \in \mathcal{DB}(G)$. Then $C_S(G)=S'$ if and only if $C_{S'}(G)=S$.\label{css'}
\end{theorem}
\begin{proof}
Suppose $C_S(G)=S'$. Since $S \cup N(S) =V$, $C_S(G)=\overline{S^c}$. That is $S'=\overline{S^c}$. For every $x \in S^c$, $e_{\overline{S^c}}(x)=diam(G)$. Since $G$ is unique eccentric vertex graph and $S$ is a boundary set, for every $x \in S$, $e_{{\overline{S^c}}}(x)=diam(G)-1$. Hence $C_{S'}(G)=C_{\overline{S^c}}(G)=S$. Conversely assume $C_{S'}(G)=S$. To prove $C_S(G)=S'$. Since $C_S(G)=\overline{S^c}$ we need only prove that $S'=\overline{S^c}$. Let $x \in S'$. If $x \in \overline{S}$ then $x =\bar{y}$ where $y \in S$.  Then we have $d(x,y)=diam(G)$. Since $S$ is the $S'$-center of $G$ this implies $C_S'(G)=V$. But this contradicts the fact that $S$ is a boundary set. Hence $x \in \overline{S^c}$ or $S' \subseteq \overline {S^c}$. Now to prove that $\overline{S^c} \subseteq S'$. On the contrary assume that there exists an $x \in \overline{S^c}$ such that $x \notin S'$. Let $x=\overline{y}$ where $y \in S^c$. Since the eccentric vertex of $y$, $x$, does not belong to $S'$, $e_{S'}(y) \le diam G-1$. If $z \in S'$ then $z \in \overline{S^c}$. Let $z=\overline{w}$ where $w \in S^c$. Since $S \cup N(S)=V$ there exists a $w'$ adjacent to $w$ such that $w'$ belong to $S$.  We have $e_{S'}(w')=diam G-1$. This implies $y \in S$, contradicting the choice of $y$. Therefore $S'= \overline{S^c}$.
\end{proof}
\begin{theorem}
Let $G$ be a symmetric even graph. Then
\begin{enumerate}[i)]
\item $S \in \mathcal{DB}(G)$ if and only if $C_S(G) \in \mathcal{DB}(G)$.
\item For $S_1,S_2 \in \mathcal{DB}(G)$, $C_{S_1}(G)=S_2$ if and only if $C_{S_2}(G)=S_1$.
\end{enumerate}
\end{theorem}
\vspace{-4mm}
\begin{proof}
\begin{enumerate}[i)]
\item Suppose $S\subseteq V$ is such that $S \in \mathcal{DB}(G)$ and let $S'=C_S(G)$. Since $S'$ is a center set of a symmetric even graph if and only if it is a boundary set, to prove that $S' \in \mathcal{DB}(G)$ we need only prove that $S' \cup N(S')=V$.  Since $S \cup N(S)=V$, $S'=\overline{S^c}$. Let $x \notin S'$. Therefore $x \in \overline {S}$ since the graph is symmetric even. Let $x = \bar{y}$ where $y \in S$. Since $S$ is a boundary set there exists a vertex $y'$ adjacent to $y$ such that $y' \in S^c$. We have $d(x,y')=diam(G)-1$. Since $G$ is symmetric even there exists a vertex $x'$ adjacent to $x$ such that $d(x',y')=diam(G)$. That is $x' \in \overline {S^c}$ or $x' \in S'$. In other words $x \in N(S')$. Hence $S' \cup N(S')=V$. Conversely suppose $S' \subseteq V$ is such that $S' \in \mathcal{DB}(G)$ and $C_S(G)=S'$ for an $S' \subseteq V$. To prove $S \in \mathcal{DB}(G)$. By the previous theorem $C_S(G)=S'$ implies $C_{S'}(G)=S$.  Now $S' \subseteq V$ is such that $S' \in \mathcal{DB}$ and $C_{S'}(G)=S$ and hence as proved earlier we can prove that $S\cup N(S)=V$ or  $S \in \mathcal{DB}(G)$.
\item This part is obvious from Theorem \ref{css'}.
\end{enumerate}
\end{proof}

\section{Enumerating Center Sets}\label{enumeration}
In designing and modelling networks it is important to have more center sets to locate facilities. Therefore the number of center sets is a good indication to the structural well-behavedness of the graph. In this section we enumerate the center sets of various classes of graphs. We first give the following definition.
The number of distinct center sets of a graph $G$ is defined as the \textit{Center number} of $G$ and is denoted by $cn(G)$.
The following results gives the center numbers of some familiar classes of graphs. The proofs of the Lemma \ref{cnkn} to Lemma \ref{cnwn}  follows from the Corollary~\ref{cor1}, Theorem~\ref{kmn}, Corollary~\ref{cor2}, Theorem~\ref{kn-e} and Proposition~\ref{wn}  respectively, so we leave the proofs.
\begin{lemma}\label{cnkn}
The center number $cn(G)$ when $G=K_n$ is $n+1$.
\end{lemma}
\begin{lemma}
The center number $cn(K_{m,n})$ is $m+n+3$ where $m,n>1$.
\end{lemma}
\begin{lemma}
The center number $cn(T)$ of a tree $T$ on $n$ vertices is $2n-1$.
\end{lemma}
\begin{lemma}
For the graph $K_n-e$ where $e \in E$, $cn(K_n-e)=n+4$.
\end{lemma}

\begin{lemma}\label{cnwn}
For the wheel graph $W_n$,
\begin{eqnarray*}
 cn(W_n) & = & 4n-3 \text{ if } n \ge 6\\
           & = & 4n-1 \text{ if } n=5
           \end{eqnarray*}
\end{lemma}

We now determine the center number of odd and even cycles. For that we introduce the following terms.
 Suppose we have $n$ \textit{linearly arranged} objects. Let $L(n,k)$ denote the number of ways of choosing $k$ objects from these $n$ objects so that no three consecutive objects are simultaneously chosen. let $L_1(n,k)$ denote the number ways to choose $k$ objects from these $n$ objects so that no two objects from alternate positions are simultaneously chosen and let $L_2(n,k)$ denote the number of ways to choose $k$ objects from these $n$ objects so that no two objects from consecutive positions are simultaneously chosen.

  Consider $n$ \textit{circularly arranged} objects where $n \ge 4$. Let $R(n,k)$ denote the number of ways to choose $k$ objects from these $n$ objects so that three objects from three consecutive positions are not chosen and $R_1(n,k)$ denote the number ways to choose $k$ objects from these $n$ objects so that no two objects from alternate positions are simultaneously chosen.
 Here we assume $n \ge 4$ since we are interested only in cycles of length greater than 3.
\begin{lemma}
 $L(n,k)=\tbinom{n}{k} \tbinom{n-k+1}{0}-\tbinom{n-3}{k-3}\tbinom{n-k+1}{1}+\tbinom{n-6}{k-6} \tbinom{n-k+2}{2}-\tbinom{n-9}{k-9}\tbinom{n-k+1}{3}+\cdots$. \label{lnk}
\end{lemma}
\begin{proof}
Similar to what we did in Corollaries \ref{oddcent} and \ref{evencent}, a particular choice of $k$ objects from $n$ objects can be represented by a binary string of size $n$ where a $1$ at the $i^{th}$ position indicates that the $i^{th}$ object is chosen and a $0$ at the $j^{th}$ position indicates that the $j^{th}$ object is not chosen. 
 So the number of choices of the required type is actually the number of binary strings of size $n$ having $k$ $1$'s and  not containing three consecutive $1$'s. Let $x_0$ denote the number of $1$'s before the first 0, for $1 \le i \le n-k-1$, let $x_i$ denote the number of 1's between the $i^{th}$ 0 and the $(i+1) ^{th}$ 0 and let $x_{n-k}$ denote the number of 1's after the $(n-k)^{th}$ 0. Therefore the total number of $1$'s in a binary string is $x_0+x_1+\cdots+x_{n-k}$. 
For a binary string of our choice, $0 \le x_i \le 2$. Hence $L_1(n,k)$ is the number of different solutions of the equation
\begin{equation}
x_0+x_1+\cdots+x_{n-k}=k, 0 \le x_i \le 2 \label{eq1}
\end{equation}

Now consider the product
\begin{equation}
\underbrace{(1+t+t^2) \times  \cdots \times (1+t+t^2)}_{(n-k+1)\text{ times}} \label{eq2}
\end{equation}
 In the expansion of this product, taking $t^{y_0}$ from the first term, $t^{y_1}$ from the second term, $\ldots$, $t^{y_{n-k}}$ from the $n-k+1^ { th}$ term we get $t^{y_0+y_1+\ldots+ t_{n-k}}$. Therefore any solution of the equation
 \begin{equation}
 y_0+y_1+\ldots+y_{n-k}=k, 0 \le y_i \le 2 \label{eq3}
 \end{equation}
  gives us the term $y^k$ in the expansion. In other words the number of solutions of equation \ref{eq3} is the coefficient of $t^k$ in expression \ref{eq2}. Since the Equations \ref{eq1} and \ref{eq3} are same, we get that $L(n,k)$ is the coefficient of $t^k$ in $(1+t+t^2)^{n-k+1}$.
\begin{align*}
(1+t+t^2)^{n-k+1} &=  \left(\frac{1-t^3}{1-t}\right)^{n-k+1}\\
                  &= (1-t^3)^{n-k+1} (1-t)^{-(n-k+1)}\\
                  &=  \left(1-\tbinom{n-k+1}{1} t^3+ \tbinom{n-k+1}{2}t^6+ \cdots \right)\\
                  &      ~~~~~~~~~~~~~ \times \left(1+\tbinom{n-k+1}{1}t+\tbinom{n-k+2}{2}t^2+\cdots+ \tbinom{n}{k}t^k+\cdots\right)
  \end{align*}
Therefore $L(n,k)$= 
 $\tbinom{n}{k} \tbinom{n-k+1}{0}-\tbinom{n-3}{k-3}\tbinom{n-k+1}{1}+\tbinom{n-6}{k-6} \tbinom{n-k+2}{2}-\tbinom{n-9}{k-9}\tbinom{n-k+1}{3}+\cdots$.\\
 The series on the right hand side is finite as all the terms after a finite number of terms shall be zero.
\end{proof}
\begin{lemma}
$R(n,k)=L(n-1,k)+2L(n-4,k-2)+L(n-3,k-1)$, $n \ge 4,~ k \ge 2$.
\end{lemma}
\begin{proof}
Let the $n$ circularly arranged objects be $v_1,\ldots,v_n$. The set of all objects such that no 3 objects from 3 consecutive positions are chosen can be divided in to the following types
\begin{enumerate}[Type I:]
\item The object $v_n$ is chosen and the objects $v_{n-1}$ and $v_1$ are not chosen. Then the total number of choices is $L(n-3,k-1)$. (See Figure 1)
\begin{figure}[htb]
\psset{unit=.75}\label{fig1}
\begin{pspicture}(3,3)(-7.9,-1)
    \psdot(1,1)
 \psdot[dotstyle=o,dotsize=12pt](2,2)
 \psdot(2,2)
  \psdot(3,3)
  \uput[u](3,3){$v_n$}
   \psdot(4,2)
   \psdot[dotstyle=o,dotsize=12pt](4,2)
 \psdot(4,2)
  \uput[r](4.1,2){$v_1$}
 \psdot(2,2)
  \uput[r](2.1,2){$v_{n-1}$}
    \psdot(5,1)
    \uput[r](5,1){$v_2$}
    \uput[r](1,1){$v_{n-2}$}
     \psdot(5,-1)
     \psdot(1,-1)
\uput[u](3,-1){$L(n-3,k-1)$}
\psdot(3,-1)
   \psline[](1,1)(2,2)(3,3)(4,2)(5,1)
   \psline[linestyle=dotted,dotsep=2pt](5,1)(5,-1)(1,-1)(1,1)
\end{pspicture}
\caption{~}
\end{figure}


\item The objects $v_n$ and $v_{n-1}$ are chosen and $v_1$ is not chosen. $v_n$ and $v_{n-1}$ are chosen implies $v_{n-2}$ is not chosen. In this case the number of choices is $L(n-4,k-2)$.
\item The objects $v_n$ and $v_{1}$ are chosen and $v_{n-1}$ is not chosen. Again as in the previous case the total number of choices is $L(n-4,k-2)$.
\item The object $v_n$ is not chosen. Here the total number of choices is $L(n-1,k)$.
\end{enumerate}
Therefore $R(n,k)=L(n-1,k)+2L(n-4,k-2)+L(n-3,k-1)$.\\
\end{proof}
It is obvious that
\begin{eqnarray*}
R(n,k)& = & 1 \text{ when }k=0\\
      & = & n \text{ when }k=1
      \end{eqnarray*}
      Now we have determined $R(n,k)$ for all $n \ge 4$ and $k \ge 0$.
\begin{theorem}
The center number of the even cycle $C_{2n}$ is $\sum\limits_{k=1} ^ {\lfloor\frac{4n}{3}\rfloor} R(2n,k)+1$.
\end{theorem}
\begin{proof}
By the Corollary \ref{evencent}, the maximum cardinality among the center sets other than $V$ is $\lfloor\frac{4n}{3}\rfloor$ and by the Theorem \ref{symevcent}, $R(2n,k)$ gives the number of center sets of size $k$ where $k \le \lfloor\frac{4n}{3}\rfloor$. Also $V$ is a center set. Hence $cn(C_{2n})=\sum\limits_{k=1}^{\lfloor\frac{4n}{3}\rfloor} R(2n,k)+1$.
\end{proof}
Before proving the center number of odd cycles, we prove the following lemmata.
We first find $L_2(n,k)$ for  given values of $n$ and $k$.
 \begin{lemma}
 $L_2(n,k)=\tbinom{n-k+1}{k}$.\label{l2nk}
 \end{lemma}
 \begin{proof}
 As in Lemma~\ref{lnk}, we give a binary representation for a particular choice of $k$ objects that conforms to the conditions specified in the definition of $L_2(n,k)$. For each $1$ in this binary representation we count the total number of $0$'s preceding this $1$. So if we have $k$ $1$'s then we get $k$ numbers from $\{0,1,\ldots,n-k\}$ and all these are distinct since there should be atleast one $0$ between any two successive $1$'s. Thus corresponding to each choice of $k$ objects of the desired type we get a unique set of $k$ distinct numbers from $\{0,1,\ldots,n-k\}$. Conversely each choice of $k$ distinct numbers from $\{0,1,\ldots,n-k\}$ gives us a unique choice of $k$ objects from $n$ linearly arranged objects satisfying the specified condition. Thus we get a one-to-one correspondence between the $k$-element subsets of $\{0,1,\ldots,n-k\}$ and the choices of $k$ objects as specified in the definition of $L_2(n,k)$.  Hence $L_2(n,k)=\tbinom{n-k+1}{k}$.
 \end{proof}
 \begin{lemma}
 $L_1(n,k)= \sum\limits_{\ell=0}^{k}L_2(\lfloor\frac{n}{2}\rfloor,\ell)L_2(\lceil\frac{n}{2}\rceil,k-\ell)$.  \label{l1nk}
 \end{lemma}
 \begin{proof}

Consider $n$ linearly arranged objects. Choosing $k$ objects from these $n$ objects  such that no two objects are from alternate positions can be done as follows. First choose  $\ell$ objects from $\lceil\frac{n}{2}\rceil$ objects in the odd positions such that no two objects are consecutive among these $\lceil\frac{n}{2}\rceil$ objects.  This can be done in $L_2(\lceil\frac{n}{2}\rceil,\ell)$ ways.  Now choose $k-\ell$ objects from the remaining  $\lfloor\frac{n}{2}\rfloor$ objects in the even positions, such that no two objects are consecutive among these $\lfloor\frac{n}{2}\rfloor$ objects.  This can be done $L_2(\lfloor\frac{n}{2}\rfloor ,k-\ell)$ ways.  Hence
 $L_1(n,k)= \sum\limits_{\ell=0}^{k}L_2(\lceil\frac{n}{2}\rceil,\ell)L_2(\lfloor\frac{n}{2}\rfloor,k-\ell)$.
 \end{proof}

\begin{lemma}
$L_1(n,k)=\sum\limits_{\ell=0}^{k}\tbinom{\lfloor\frac{n}{2}\rfloor-\ell+1}{\ell}\tbinom{\lceil\frac{n}{2}\rceil-(k-\ell)+1}{k-\ell}$.
\end{lemma}
\begin{proof}
The proof follows from Lemma \ref{l1nk} and Lemma \ref{l2nk}.
\end{proof}
  \begin{lemma}\label{c1nk}
 $R_1(n,k)=L_1(n-2,k)+2L_1(n-5,k-1)+3L_1(n-6,k-2)$, $n \ge 6, k \ge 2$.
 \end{lemma}
 \begin{proof}
Let the  $n$ circularly arranged objects be $v_1, \ldots,v_n$. The set of all choices of $k$ objects such that no two objects occupy alternate positions can be divided in to various types.
\begin{enumerate}[Type I:]
\item Both $v_n$ and $v_{n-1}$ are not chosen. In this case the total number of choices is $L_1(n-2,k)$ (See Figure 2).

\begin{figure}[H]
\psset{unit=.8}\label{fig2}
\begin{pspicture}(3,3)(-7.9,-1)
    \psdot(1,1)
 \psdot[dotstyle=o,dotsize=12pt](2,2)
 \psdot(2,2)
  \psdot(3,3)
   \psdot(4,2)
   \psdot[dotstyle=o,dotsize=12pt](3,3)
   \psdot(3,3)
  \uput[u](3,3.1){$v_n$}
 \psdot(4,2)
  \uput[r](4.1,2){$v_1$}
 \psdot(2,2)
  \uput[r](2.1,2){$v_{n-1}$}
    \psdot(5,1)
    \uput[r](5,1){$v_2$}
    \uput[r](1,1){$v_{n-2}$}
     \psdot(5,-1)
     \psdot(1,-1)
\uput[u](3,-1){$L_1(n-2,k)$}
   \psline[](1,1)(2,2)(3,3)(4,2)
   \psline[linestyle=dotted,dotsep=2pt](4,2)(5,1)(5,-1)(1,-1)(1,1)
\end{pspicture}
\caption{}
\end{figure}


\item $v_n$ is selected and $v_{n-1}$ is not selected. $v_n$ is selected implies $v_{n-2}$ and $v_2$ are not selected. The number of choices where $v_1$ is selected is $L_1(n-6,k-2)$ and the number of choices where $v_1$ is not selected is $L_1(n-5,k-1)$. Hence the total number of such choices is $L_1(n-6,k-2)+L_1(n-5,k-1)$.
\item $v_n$ is not selected and $v_{n-1}$ is selected. As in the previous case the total number of such choices is $L_1(n-6,k-2)+L_1(n-5,k-1)$.
\item Both $v_n$ and $v_{n-1} $are selected. $v_n$ and $v_{n-1} $ are selected implies $v_1,v_2,v_{n-2}$ and $v_{n-3}$ are not selected. Therefore the number of choices of this type is $L_1(n-6,k-2)$.
\end {enumerate}
Hence $R_1(n,k)=L_1(n-2,k)+2L_1(n-5,k-1)+3L_1(n-6,k-2)$.
\end{proof}
Now it is easy to see that
\begin{eqnarray*}
R_1(n,k) & = & 1,  \text{ when } k=0\\
          & = & n, \text{ when } k =1\text{ or }k=2\text{ and }n=4\text{ or }5\\
                 & = & 0,  \text{ when } k \ge 3, ~ n=4\text{ or }5 \\
          \end{eqnarray*}
 Thus we have determined $R_1(n,k)$ for all $n \ge 4$ and $k \ge 0$.

Now with the help of   Theorem~\ref{odd} and Corollary~\ref{oddcent}, we have the center number of the odd cycle $C_{2n+1}$, $n \ge 2$.

\begin{theorem}
The center number of the odd cycle $C_{2n+1}$, $n \ge 2$, is $\sum\limits_{k=1} ^ {n} R_1(2n+1,k)+1$.
\end{theorem}
\section{Conclusion}
In this paper we could identify and characterize the center sets of many classes of graphs. 
We proved that the induced subgraphs of center sets of block graphs, wheel graphs etc are all connected. Characterizing the class of graphs for which the induced subgraphs of all center sets are connected is another problem that will be interesting. We also determined the center number of some classes of graphs, which also need to be extended. 
\bibliographystyle{plain}
\bibliography{ram-biblio}

\end{document}